\documentclass[letterpaper, 10 pt, conference]{ieeeconf}  

\IEEEoverridecommandlockouts                              
\usepackage{graphics} 
\usepackage{epsfig} 
\usepackage{times} 
\usepackage{amsmath} 
\usepackage{amssymb}  

\usepackage{graphicx}
\usepackage{mathtools}
\usepackage{amsfonts}
\usepackage{xcolor}
\usepackage{soul}
\usepackage{subcaption}
\usepackage{bm}
\usepackage{algorithm}
\usepackage{algpseudocode}
\usepackage{cite}

\def\RR{\mathbb{R}}
\def\SS{\mathbb{S}}
\def\BB{\mathbb{B}}
\def\setA{\mathcal{A}}
\def\calA{\mathcal{A}}
\def\calB{\mathcal{B}}
\def\calC{\mathcal{C}}
\def\calD{\mathcal{D}}
\def\calK{\mathcal{K}}
\def\calL{\mathcal{L}}

\def\tt{^{\mathsf{\scriptscriptstyle T}}}
\def\x{{\rm x}}
\def\aa{{\rm a}}
\def\a{{\rm a}}
\def\ua{{u_\aa}}

\def\z{{\rm z}}
\def\e{{\rm e}}
\def\calS{\mathcal{S}}
\def\g{{\rm g}}
\def\v{{\rm v}}
\def\w{{\rm w}}
\def\str{^{\star}}

\DeclareMathOperator{\signof}{sign}

\newcommand{\bbm}[1]{\begin{bmatrix} #1 \end{bmatrix}}

\newcommand{\matl}[1]{\begin{matrix*}[l] #1 \end{matrix*}}

\newcommand\norm[1]{\left\lVert#1\right\rVert}

\newcommand{\dintt}[4]{\int_{#1}^{#2}#3\,\mathrm{d}#4}

\newtheorem{theorem}{Theorem}
\newtheorem{lemma}{Lemma}
\newtheorem{corollary}{Corollary}
\newtheorem{definition}{Definition}
\newtheorem{remark}{Remark}
\newtheorem{assumption}{Assumption}

\definecolor{myblue}{RGB}{0,144,178}
\definecolor{myyellow}{RGB}{230,159,0}
\definecolor{mygreen}{RGB}{0,158,115}
\definecolor{myred}{RGB}{213,94,0}
\definecolor{mymaroon}{RGB}{140,29,64}
\definecolor{gravitygreen}{RGB}{19,137,31}
\definecolor{velocityblue}{RGB}{28,93,207}
\definecolor{thrustred}{RGB}{179,28,53}
\definecolor{myhighlight}{RGB}{255,230,150}

\sethlcolor{myhighlight}

\newcommand{\saveforjournal}[1]{}

\title{\LARGE \bf
Practical Complete Tracking With Pivoted Unidirectional Actuation%
}

\author{Ian J. Willebeek-LeMair and Craig A. Woolsey
\thanks{This work was supported in part by The Boeing Company.
}
\thanks{I. Willebeek-LeMair is a Ph.D. student in the Kevin T. Crofton Department of Aerospace and Ocean Engineering, 
        Virginia Tech, Blacksburg, VA 24061, USA
        {\texttt{\small ianwl@vt.edu}}}%
\thanks{C. Woolsey is a Professor in the Kevin T. Crofton Department of Aerospace and Ocean Engineering, 
        Virginia Tech, Blacksburg, VA 24061, USA
        {\texttt{\small cwoolsey@vt.edu}}}%
}

\begin{document}

\maketitle
\thispagestyle{empty}
\pagestyle{empty}

\begin{abstract}

This paper addresses practical tracking control for robotic vehicles with pivoted unidirectional actuators.  A widely recognized obstacle to tracking control for multirotors is the free-fall singularity --- when commanded to free-fall, the vehicle's attitude target becomes undefined.  This singularity is a pathology pertinent to all robotic vehicles with pivoted and gimbaled unidirectional actuators.  This paper addresses the planar instance of the general theoretical problem: the case of a vehicle with pivoted unidirectional actuation.  Starting from a baseline robust controller that assumes unconstrained inputs, we redesign the control law to be compatible with the pivoted actuator.  This is accomplished by driving the output of the pivoted actuator to a ball centered at the target input value.  The baseline controller's guarantees are recovered in a practical sense.  The theory is illustrated with a simulation example.  

\end{abstract}

\section{INTRODUCTION} \label{sec:Introduction}

The theoretical problem addressed in this paper is motivated by the applied problem of position and velocity tracking for a class of robotic vehicles that includes multirotor unmanned aerial vehicles (UAVs). State-of-the-art nonlinear controllers for multirotors are often designed using backstepping-based approaches that leverage the geometric relationship between a multirotor's orientation and the direction of its thrust vector to synthesize a target attitude from a desired thrust history~\cite{lee_geometric_2010,lee_control_2011,lee_geometric_2011,lee_nonlinear_2013,gamagedara_geometric_2019,willebeek_lemair_input_state_2025,willebeek_lemair_robust_2025,flores_robust_2025,naldi_robust_2017,martins_global_2021,martins_global_2024}.  This target attitude becomes undefined whenever the desired thrust goes through zero, i.e. when the desired motion is free-fall, because the vehicle can free-fall in any orientation; see Fig.~\ref{fig:state_of_the_art}.  The free-fall singularity obstructs the development of robust and adaptive controllers for all robotic vehicles that employ pivoted or gimbaled unidirectional actuators.  

Recent multirotor control literature addresses the free-fall singularity by either assuming it is never encountered~\cite{lee_geometric_2010,lee_control_2011,lee_geometric_2011,lee_nonlinear_2013,gamagedara_geometric_2019,willebeek_lemair_input_state_2025,willebeek_lemair_robust_2025,flores_robust_2025} or excluding any desired position history with a vertical acceleration faster than gravity~\cite{naldi_robust_2017,martins_global_2021,martins_global_2024}. These assumptions and exclusions limit theoretical results.   For example, assuming the free-fall singularity is never encountered complicates input-to-state stability (ISS) analysis: a disturbance could push the vehicle through the singularity, causing a jump discontinuity in the ISS-Lyapunov function. Similarly, limiting allowable acceleration to be less than gravity precludes ISS results since disturbances might overpower the controller.  

Gimbaled and pivoted unidirectional actuators are often treated as rate-unlimited inputs, even though actuator rotation takes time. Standard backstepping approaches account for the pivot subsystem in the same manner as is used for multirotors and so suffer the same pathology: a desired actuation that goes through zero can demand a discontinuous jump in pivot direction. Inspired by the multirotor scenario, we continue to refer to this singularity as the \textit{free-fall singularity}.  This paper describes a controller redesign that overcomes the singularity and enables tracking without artificial actuation limits.  

\begin{figure}[t!]
    \centering
    \begin{minipage}{0.4865\textwidth}
        \centering
        \includegraphics[trim={0.0cm 0cm 0.0cm -0.4cm},clip,width=0.99\textwidth]{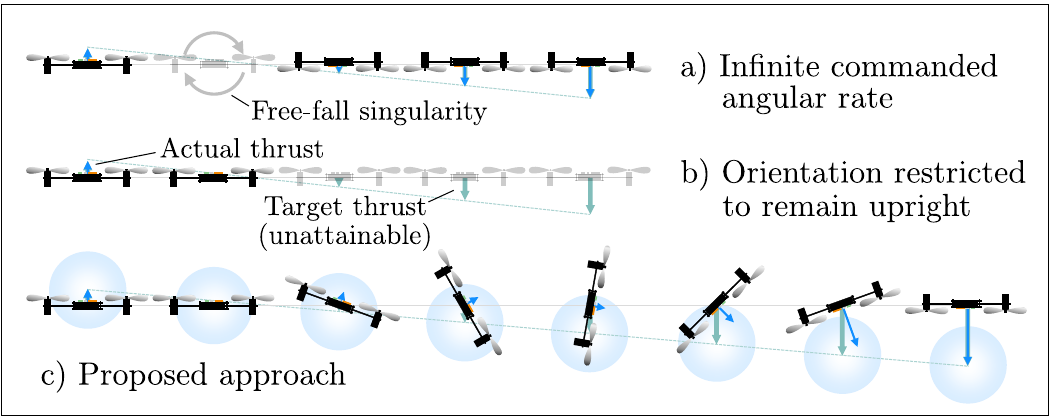}
        \caption{Three multirotor command tracking scenarios.}
        \label{fig:state_of_the_art}
    \end{minipage}
    \vspace{-0.25in}
\end{figure}

Many robotic vehicles employ or themselves are gimbaled or pivoted unidirectional actuators.  Ocean vessels that employ azimuthal pivots to direct thrust about a vertical axis are one example.  Vehicles that reorient themselves to direct a body-fixed thruster can be considered gimbaled or pivoted unidirectional actuators, as well.  Examples include multirotor UAVs, rockets, hovercraft that steer with differential thrust, and wheeled ground vehicles constrained to forward motion~\cite{lee_geometric_2010,ballaben_towards_2025,bahati_control_2025}.  In general, a system with a pivoted or gimbaled unidirectional actuator is one where a unidirectional control input enters the system as the coefficient of a unit vector whose direction and angular velocity are a controllable integrator cascade. For brevity, this paper exclusively examines the lowest dimensional case of interest: a pivoted unidirectional actuator; i.e., one for which the unidirectional actuation can be continuously rotated within a plane.  (A one-dimensional unit vector cannot change direction continuously.)  Extensions to higher dimensional gimbaled cases can be constructed using the proposed solution to the pivoted case. 

In our analysis, we partition a robotic vehicle system into an outer-loop vehicle subsystem whose state may include position, velocity, or both and an inner-loop pivot subsystem whose state is the pivot's direction and angular velocity.  We assume that the pivoted actuator can exert arbitrarily large positive actuation, that an input-to-state stabilizing exact complete tracking (ECT) controller is available for the outer-loop system, and that the ECT controller's design regards the pivoted actuator as an unconstrained input.  The controller is then redesigned to be compatible with the constrained motion of a pivot and to enable practical complete tracking (PCT) for the outer-loop vehicle system. Exact versus practical complete tracking is clarified in Definition~\ref{def:ECT_and_PCT}.  

There are two primary contributions in this paper.  The first contribution is an example of a smooth outer-loop trajectory that cannot be followed exactly by vehicles with pivoted unidirectional actuators, proving by contradiction that no ECT controller exists for this class of systems.  The second and main contribution is the synthesis of a PCT controller that ensures the control objective is satisfied. 

The paper is organized as follows.  Section~\ref{sec:problem_formulation} details the class of systems considered, states the objective, and proves no ECT controllers exist.  Section~\ref{sec:controller_design} includes the synthesis of the PCT controller and proofs of its properties.  Section~\ref{sec:Simulation} provides an example.  Section~\ref{sec:Conclusions} states conclusions.  

The following notation is used throughout the paper.  Let $I$ denote an identity matrix, $0$ denote a matrix of zeros, and $\e_i$ denote a unit vector with a one in its $i$\textsuperscript{th} entry. The dimensions of $I$, $0$, and $\e_i$ are implied by context.  Let $\calS$ denote the skew-symmetric rotation matrix $\calS=\e_2 \e_1\tt-\e_1 \e_2\tt\in\RR^{2\times2}$. Let $\norm{\cdot}$ denote the Euclidean norm.  Let $\BB^n$ denote the unit ball in $\RR^n$ with respect to $\|\cdot\|$. For a vector $x\in\RR^n$ and a set $A\subseteq\RR^n$, let $\|x\|_A=\inf_{\z\in A}\|x-\z\|$.  Let $\bar A$, $\partial A$, and $\mathring{A}$ denote the closure, the boundary, and the interior of $A$, respectively.  Let $\bar\RR = \RR\cup\{-\infty,\infty\}$ denote the extended real number line.  For any $f:\RR^n\to\RR^m$, let $f(A)=\{v\in\RR^m|v=f(a),a\in A\}$.  Let $M\otimes N$ denote the Kronecker product of $M$ and $N$.  Finally, let $C^r(\mathbb{X},\mathbb{Y})$ and ${\rm PWC}(\mathbb{X},\mathbb{Y})$ denote, respectively, the set of $r$-times continuously differentiable functions and the set of piecewise continuous functions from $\mathbb{X}$ to $\mathbb{Y}$.

\section{PROBLEM FORMULATION} \label{sec:problem_formulation}

Consider the dynamical system \vspace{-0.125cm}
\begin{subequations} \label{eq:olsys}
    \begin{align}
            \dot x(t) &= f(t, x(t)) + \g(t, x(t))\lambda(t)u_1(t)
            \label{eq:olsys_outer_loop_system}
            \\
            \dot\lambda(t) &= \calS \lambda(t) \omega(t) 
            \label{eq:olsys_pivot_pose}
            \\
            \dot\omega(t) &= u_2(t) 
            \label{eq:olsys_pivot_rate} 
    \end{align}
\end{subequations}
where $ x\in C^{0}(\RR_{\geq0},\RR^n)$ is referred to as the state of the outer-loop system, $f\in C^{\infty}(\RR_{\geq0}\times\RR^n,\RR^n)$ is the smooth drift vector field, $\g\in C^{\infty}(\RR_{\geq0}\times\RR^n,\RR^{n\times2})$ is the smooth actuation matrix which is assumed to be bounded and uniformly full column rank, $\lambda\in C^1(\RR_{\geq0},\SS^1)$ is the direction of the pivoted unidirectional actuator, $\omega\in C^0(\RR_{\geq0},\RR)$ is the angular velocity of the pivot, $u_1\in {\rm PWC}(\RR_{\geq0},\RR_{\geq0})$ is the unidirectional input to the outer loop, and $u_2\in {\rm PWC}(\RR_{\geq0},\RR)$ is the angular acceleration input to the pivot. 

In this paper,~\eqref{eq:olsys} is interpreted as the interconnection of an affine-in-control outer-loop subsystem~\eqref{eq:olsys_outer_loop_system} with an inner-loop pivot subsystem~\eqref{eq:olsys_pivot_pose},~\eqref{eq:olsys_pivot_rate}.  The effective input to the outer-loop system~\eqref{eq:olsys_outer_loop_system} is \vspace{-0.17cm}
\begin{equation} \label{eq:defn_of_gimbal_thrust} 
    \aa(t) = \lambda(t)u_1(t) \in \RR^2,
    \vspace{-0.17cm}
\end{equation}
which we refer to as the ``actuation'' or ``thrust'' applied by the pivot.  The control-affine structure of~\eqref{eq:olsys_outer_loop_system} and the requirements that $\g$ be bounded and uniformly full column rank are important; they hold for all the vehicle systems mentioned in Section~\ref{sec:Introduction}.  The requirements on $\g$ ensure that the sensitivity of $\dot x(t)$ to changes in $\aa(t)$ is always both uniformly bounded and nonzero. Boundedness implies small changes in $\aa(t)$ yield proportionally small changes in $\dot x(t)$, and being nonzero implies that $\aa(t)$ can always influence $\dot x(t)$.  

To develop control laws for~\eqref{eq:olsys}, actuation $\aa(t) = \lambda(t)u_1(t)$ is often replaced with an unconstrained input $u_\aa(t)\in\RR^2$ which leads to the fictitious surrogate system\vspace{-0.2cm}
\begin{equation}
    \dot x(t) = f(t, x(t)) + \g(t, x(t))\ua(t) + \w(t)
    \label{eq:surrogate_olsys} \vspace{-0.175cm}
\end{equation}
where $\ua:\RR_{\geq0}\to\RR^2$ is the (fictitiously) unconstrained actuation input and $\w:\RR_{\geq0}\to\RR^n$ is a bounded disturbance.  The system~\eqref{eq:surrogate_olsys} is a surrogate for the entirety of~\eqref{eq:olsys} (not just~\eqref{eq:olsys_outer_loop_system}).  The mock input $u_\aa(t)$ suffers none of the consequences of the constraints $\lambda(t)\in\SS^1$ and $u_1(t)\geq0$.  These effects are instead captured by $\w(t)$.

Let $X\str$ denote the set of undisturbed ($\w\equiv0$) solutions of~\eqref{eq:surrogate_olsys} arising from initial conditions $x(0)\in\RR^n$ and smooth input histories $u_\aa \in C^\infty(\RR_{\geq0},\RR^2)$.  Then for any smooth desired trajectory $ x\str\in X\str$,~\eqref{eq:surrogate_olsys} can be rewritten \vspace{-0.15cm}
\begin{equation}
    \dot e(t) = f_{ x\str}(t,e(t)) + \g_{ x\str}(t,e(t))u_\aa(t) + \w(t)
    \label{eq:error_surrogate_olsys} \vspace{-0.15cm}
\end{equation}
where $e(t)= x(t)- x\str(t)$, $f_{ x\str}(t,e(t))=f(t,e(t)+ x\str(t))-\dot x\str(t)$, and $\g_{ x\str}(t,e(t))=\g(t,e(t)+ x\str(t))$.

\begin{assumption}
\label{assmp:ISS_control_law}
    Consider the fictitious system~\eqref{eq:error_surrogate_olsys}. Assume there exists a known control law $\kappa_{ x\str}\in C^{\infty}(\RR_{\geq0}\times\RR^n,\RR^2)$ such that $u_\aa(t)=\kappa_{ x\str}(t,e(t))$ input-to-state stabilizes~\eqref{eq:error_surrogate_olsys} with respect to the disturbance $\w$, uniformly in $X\str$.  Then there exist $\beta\in\calK\calL$ and $\gamma\in\calK_{\infty}$ such that \vspace{-0.17cm}
    \begin{equation}
        \|e(t)\| \leq \beta(\|e(t_0)\|,t-t_0) + \gamma({\sup}_{\tau\geq t_0}\|\w(\tau)\|) \label{eq:outer_clsys_is_pre_ISS}
        \vspace{-0.15cm}
    \end{equation}
    for all $ x\str\in X\str$, for each $\w\in L^{\infty}(\RR_{\geq0},\RR^n)$.  
\end{assumption}

The objective in this paper is to synthesize a control law for the inputs $u_1$ and $u_2$ that enables the outer-loop system~\eqref{eq:olsys_outer_loop_system} to track any user-supplied smooth outer-loop trajectory in a practical sense.  

\begin{definition} \label{def:ECT_and_PCT}
    Suppose there exists a controller, parameterized by $x\str$, that ensures \vspace{-0.17cm}
    \begin{equation}
        \limsup{}_{t\to\infty}\| x(t)- x\str(t)\| \leq \Delta
        \label{eq:control_objective} \vspace{-0.15cm}
    \end{equation}
    for every $x\str\in X\str$, where $x(t)$ is a solution of~\eqref{eq:olsys_outer_loop_system}.  If~\eqref{eq:control_objective} holds with $\Delta=0$ and $x(0) = x\str(0)$ implies $x \equiv x\str$, then the controller is called an \textit{exact complete tracking} (ECT) \textit{controller}.  If~\eqref{eq:control_objective} holds with $\Delta>0$ it is called a \textit{practical complete tracking} (PCT) \textit{controller}.  (The choice of terminology is discussed in Remarks~\ref{remark:complete} and~\ref{remark:practical}.)
\end{definition}

The control objective is to redesign the ISS ECT controller $\kappa_{ x\str}$ for~\eqref{eq:surrogate_olsys} into a PCT controller for~\eqref{eq:olsys}.  The motivation to synthesize a PCT controller is that ECT controllers for~\eqref{eq:olsys} do not exist.  This fact follows from the existence of $x\str\in X\str$ that are not outer-loop trajectories of~\eqref{eq:olsys}, as formalized and discussed in the following theorem and remark. 

\begin{theorem} \label{thrm:no_ECT}
    There exist no ECT controllers for~\eqref{eq:olsys}.
\end{theorem}

\begin{proof}
    ECT requires the implication $x(0)= x\str(0) \Rightarrow x \equiv x\str$ for any $x\str\in X\str$.  If there exists a trajectory $x\str$ of~\eqref{eq:surrogate_olsys} for which this implication is impossible, then ECT is likewise impossible.  To prove the theorem, we provide such a trajectory of~\eqref{eq:surrogate_olsys} that is not an outer-loop trajectory of~\eqref{eq:olsys}. 
    
    Denote by $x\str$ a trajectory of~\eqref{eq:surrogate_olsys} that results from an arbitrary initial condition $x\str(0)\in\RR^n$ and the input history $u_\aa(t)=(1-t)\e_2$.  Consider a trajectory $x$ of~\eqref{eq:olsys} with $x(0)=x\str(0)$.  Since $\g$ is uniformly full column rank, $x\equiv x\str$ only if $\lambda(t)u_1(t) = u_\aa(t)$ for all $t\in\RR_{\geq0}$.  The pivot direction consistent with $u_\aa(t)$ must therefore satisfy \vspace{-0.17cm}
    \begin{equation}
        \lambda(t) = \e_2 \,\,\,\, \forall \, t\in[0,1), \quad \lambda(t) = -\e_2  \,\,\,\, \forall \, t\in(1,\infty)
        \label{eq:no_ECT} \vspace{-0.17cm}
    \end{equation}
    since $u_1(t)\geq0$.  Hence, $x \equiv x\str$ requires a jump discontinuity in $\lambda(t)$ at $t=1$ which contradicts $\lambda\in C^1$.  Therefore there exist no input histories $u_1,u_2$ for which $x \equiv x\str$, and so ECT is impossible.
\end{proof}

\begin{remark} \label{remark:complete}
    The use of ``complete'' in Definition~\ref{def:ECT_and_PCT} is motivated by the following observation.  Let $X$ denote the set of $x$-histories of~\eqref{eq:olsys} arising from smooth input histories.  A corollary to Theorem~\ref{thrm:no_ECT} is that $X$ is a proper subset of $X\str$.  Compared to $X\str$, the set $X$ has \textit{holes} such as the example output history provided in the proof of Theorem~\ref{thrm:no_ECT}.  ECT and PCT address reference histories drawn from the \textit{complete} set $X\str$ rather than its strict subset $X$. 
\end{remark}

\begin{remark} \label{remark:practical}
    The use of the term  ``practical'' in Definition~\ref{def:ECT_and_PCT} is consistent with \textit{input-to-state practical stability} (ISpS) where it indicates a nonzero ultimate bound~\cite{lin_various_1995}.
\end{remark}

\section{CONTROLLER DESIGN} \label{sec:controller_design}

This section details the design of the PCT control law for~\eqref{eq:olsys} with an inner-loop/outer-loop architecture based on backstepping.  Supplied with a desired outer-loop trajectory $ x\str(t)$, the outer loop provides the inner loop with a desired thrust history $\aa\str(t)=\kappa_{ x\str}(t,e(t))$.  Through backstepping, the inner loop drives $\aa(t)=\lambda(t)u_1(t)$ to the ball \vspace{-0.15cm}
\begin{equation}
    \setA(t)=\aa\str(t)+r \bar\BB^2 \vspace{-0.17cm}
\end{equation}
for some $r>0$, which is the tolerance for error in the commanded thrust vector. The set-based objective $\aa(t)\to\setA(t)$ is the core of the control design.  The difference $\aa(t)-\aa\str(t)$ is then interpreted as a disturbance to the outer-loop system and practical stabilization results follow.  The inner-loop design is described in Section~\ref{subsec:inner_loop} and the long-term behavior of the composite system is discussed in Section~\ref{subsec:long_term}.

\subsection{Inner-loop Design} \label{subsec:inner_loop}

\subsubsection{Motivating the Set-Based Approach} \label{subsubsec:why_set_based}
The thrust command $\aa\str(t)=0$ is special for a pivoted actuator.  Whenever $\aa\str(t)$ is \textit{nonzero}, there is one particular direction \vspace{-0.17cm}
\begin{equation}
    \lambda\str(t) = \frac{\aa\str(t)}{\|\aa\str(t)\|} \vspace{-0.17cm}
\end{equation}
from which the pivoted actuator can supply the target thrust.  For $\aa\str(t)=0$, in contrast, there is no unique target direction $\lambda\str(t)$ of $\lambda(t)$ --- the target $\aa\str(t)=0$ can be attained with $\lambda(t)$ pointed in \textit{any} direction.  The geometric fact that there is a \textit{set} of admissible directions $\lambda(t)$ when $\aa\str(t)=0$ motivates the set-based control objective $\aa(t)\to\setA(t)$. 

Unlike the single thrust direction $\lambda\str(t)$, the set $\setA(t)$ is defined for all values of the thrust command $\aa\str(t)$.  In particular, the ball $\setA(t)$ is well defined and centered at the origin when $\aa\str(t)\!=\!0$.  When $\|\aa\str(t)\|$ is large relative to $r$, $\aa(t)\!\in\!\setA(t)$ implies the difference between $\aa(t)$ and $\aa\str(t)$ is small relative to $\|\aa\str(t)\|$, recovering the performance of a controller designed to drive $\lambda(t)\to\lambda\str(t)$.  At the expense of inexact tracking, this set-based approach enables the controller to become independent of $\lambda\str(t)$ when $\aa\str(t)$ is in a bounded neighborhood of zero.

There are several other hurdles the inner-loop controller must overcome.  First, the topological obstruction associated with control around $\SS^1$ must be addressed~\cite{liberzon_switching_2003}.  Whenever $\lambda(t)=-\lambda\str(t)$, the control law must decide which direction to turn, given that turning in either direction would reduce the error.  To ensure the resulting control law is smooth, the obstruction is addressed by introducing a 1-dimensional analogue of modified Rodrigues parameters (MRPs) to represent the attitude of the pivot.  One-dimensional (1D) MRPs are a double covering of $\SS^1$ --- on one of the coverings, $\lambda(t)=-\lambda\str(t)$ is resolved by turning clockwise; on the other, it is resolved by turning counterclockwise.  Relying on 1D MRPs introduces another hurdle: unwinding.  Unwinding is a behavior where the vehicle turns nearly a full revolution to achieve a target attitude when a shorter turn in the opposite direction would achieve the same target attitude.  To address unwinding, the control law is initialized in a manner that discourages unwinding. 

\subsubsection{Specifying a Sufficient Target Attitude Set} \label{subsubsec:target_theta_set}
At each moment, $\setA(t)$ implicitly defines a set of admissible thrust directions $\lambda(t)$ for which there exists $u_1(t)$ such that \vspace{-0.17cm}
\begin{equation}
    \aa(t)=u_1(t)\lambda(t)\in\setA(t)
    \label{eq:geometry_0} \vspace{-0.17cm}
\end{equation}
is possible. Let $\theta(t)$ denote the signed angle from $\lambda\str(t)$ to $\lambda(t)$ whenever $\lambda\str(t)$ is defined, \vspace{-0.17cm}
\begin{equation}
    \theta(t) = \arctan_2(\lambda(t)\tt\calS\lambda\str(t),\lambda(t)\tt\lambda\str(t)).
    \label{eq:defn_of_theta} \vspace{-0.17cm}
\end{equation}
Broadly speaking, if $\lambda\str(t)$ is defined then $\lambda(t)$ is an admissible thrust direction if $|\theta(t)|$ is sufficiently small, and if $\lambda\str(t)$ is not defined then $\lambda(t)$ (and all other directions) are admissible since $\calA(t)$ encloses the origin.

For simplicity and with a significant loss of generality to be addressed in future work, prescribe \vspace{-0.17cm}
\begin{equation}
    u_1(t) = \|\a\str(t)\|. \label{eq:u1_prescription} \vspace{-0.17cm}
\end{equation}
The aim is now to identify a sufficient condition for \vspace{-0.17cm}
\begin{equation}
    \a(t)=\|\a\str(t)\|\lambda(t)\in\setA(t) \label{eq:set_membership_objective} \vspace{-0.17cm}
\end{equation}
in terms of $\theta(t)$.

To emphasize the role played by $\theta(t)$ in~\eqref{eq:set_membership_objective}, consider a coordinate frame where $\lambda\str(t)$ is aligned to the horizontal.  In this coordinate frame,~\eqref{eq:set_membership_objective} becomes \vspace{-0.17cm}
\begin{equation}
    \|\a\str(t)\| \bbm{\cos\theta(t) \\ \sin\theta(t)} \in \|\a\str(t)\| \, \e_1 + r\bar\BB^2. \label{eq:geometry_1} \vspace{-0.17cm}
\end{equation}
That is, \vspace{-0.17cm}
\begin{equation}
    \hspace{-0.1cm}\theta(t)\in\left\{\vartheta\in\bar\RR \left| \|\a\str(t)\| \bbm{\cos\vartheta \\ \sin\vartheta} \in \|\a\str(t)\| \, \e_1 + r\bar\BB^2 \right.\right\}
    \label{eq:overly_general_angle_condition} \vspace{-0.17cm}
\end{equation}
is sufficient for~\eqref{eq:set_membership_objective}.  Towards further simplifying~\eqref{eq:overly_general_angle_condition}, the set on the right hand side is rewritten as an interval through the following construction.  Let \vspace{-0.17cm}
\begin{equation}
    \calC(t) = \left\{\|\a\str(t)\| \bbm{\cos\vartheta \\ \sin\vartheta} \left|  \, \vartheta\in\bar\RR \right.\right\}
    \label{eq:defn_of_C} \vspace{-0.17cm}
\end{equation}
and \vspace{-0.17cm}
\begin{equation}
    \calD(t)=\|\a\str(t)\| \, \e_1 + r\bar\BB^2. \vspace{-0.17cm}
\end{equation}
Note $\|\a\str(t)\| [\cos\theta(t) \, \sin\theta(t)]\tt\in\calC(t)$ by construction and that~\eqref{eq:geometry_1} is equivalent to $\|\a\str(t)\| [\cos\theta(t) \, \sin\theta(t)]\tt\in\calD(t)$.
The intersection of the circle $\calC(t)$ and disk $\calD(t)$, which is connected and symmetric about the horizontal, corresponds to the angles $\theta(t)$ for which~\eqref{eq:set_membership_objective} holds.  These two sets are graphed in Figure~\ref{fig:circle_schematic} for several values of $\|\a\str(t)\|$.

\begin{figure}[t!]
    \centering
    \begin{minipage}{0.5\textwidth}
        \begin{subfigure}[t]{0.45\textwidth}
            \centering
            \includegraphics[trim={0.0cm 0cm 0.0cm -0.4cm},clip,width=0.95\textwidth]{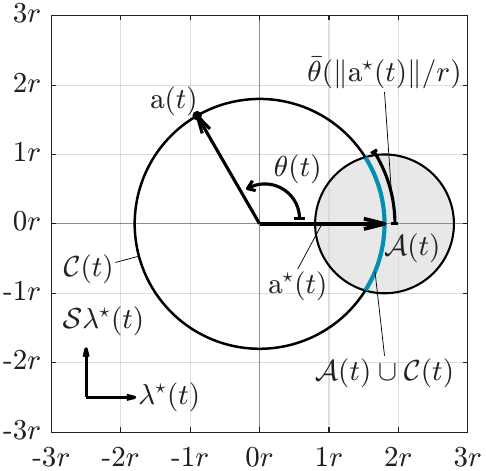}
            \caption{$\|\a\str(t)\|=1.8r$}
            \label{fig:circles_a}
        \end{subfigure}%
        ~
        \begin{subfigure}[t]{0.45\textwidth}
            \centering
            \includegraphics[trim={0.0cm 0cm 0.0cm -0.4cm},clip,width=0.95\textwidth]{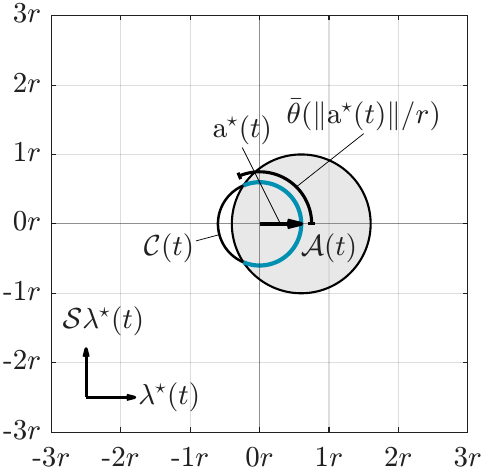}
            \caption{$\|\a\str(t)\|=0.6r$}
            \label{fig:circles_b}
        \end{subfigure}
    \end{minipage}
    \begin{minipage}{0.5\textwidth}
        \begin{subfigure}[t]{0.45\textwidth}
            \centering
            \includegraphics[width=0.95\textwidth]{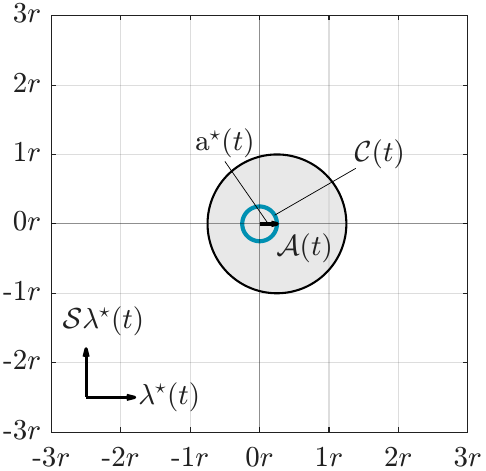}
            \caption{$\|\a\str(t)\|=0.25r$}
            \label{fig:circles_c}
        \end{subfigure}%
        ~
        \begin{subfigure}[t]{0.45\textwidth}
            \centering
            \includegraphics[width=0.95\textwidth]{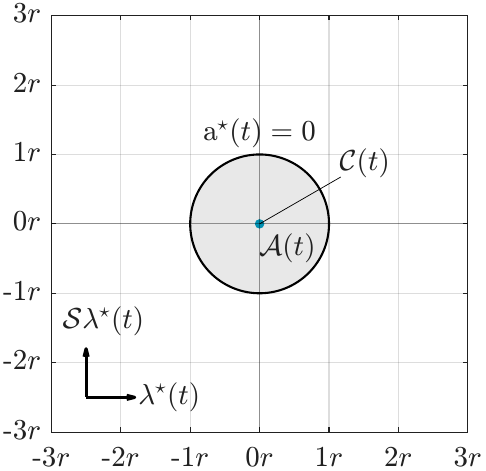}
            \caption{$\|\a\str(t)\|=0$}
            \label{fig:circles_d}
        \end{subfigure}
    \end{minipage}
    \caption{Illustrations of $\calA(t)$ (equivalently $\calD(t)$), $\calC(t)$, and their intersection, drawn with $\lambda\str(t)$ aligned to the horizontal.}
    \label{fig:circle_schematic}
    \vspace{-0.4cm}
\end{figure}

Consider the intersection of $\calC(t)$ with $\partial\calD(t)$.  
The number of intersections depends on the ratio of $\|\a\str(t)\|$ to $r$.  The circles intersect at two points when $\|\a\str(t)\|>\tfrac{1}{2}r$, at one point when $\|\a\str(t)\|=\tfrac{1}{2}r$, and at no points when $\|\a\str(t)\|<\tfrac{1}{2}r$ (because $\calC(t)$ lies strictly within the interior of $\calD(t)$).  When $\calC(t)$ and $\partial\calD(t)$ intersect, the unsigned angle from the horizontal to either line running through the origin and a point of intersection is given by \vspace{-0.17cm}
\begin{equation}
    \arccos\Bigl(1-\frac{1}{2}\frac{r^2}{\|\a\str(t)\|^2}\Bigr). \vspace{-0.17cm}
\end{equation}
This angle can be used to construct a sufficient condition for~\eqref{eq:set_membership_objective}.

\begin{lemma} \label{lemma:theta_bar}
    Define $\bar\theta:\RR_{\geq0}\to\bar\RR$ by 
    \vspace{-0.17cm}
    \begin{equation}
        \bar\theta(\sigma) = \left\{ \matl{\infty, & \sigma \leq 1/2 \\ \arccos(1-\tfrac{1}{2}\tfrac{1}{\sigma^2}), & \sigma > 1/2} \right.. \vspace{-0.17cm}
        \label{eq:defn_of_thetabar}
    \end{equation}
    Then~\eqref{eq:set_membership_objective} holds if and only if \vspace{-0.17cm}
    \begin{equation}
        \theta(t)\in \bar\Theta(t) \coloneqq [-\bar\theta(\|\a\str(t)\|/r), \bar\theta(\|\a\str(t)\|/r)].  \label{eq:geometry_2} \vspace{0.2cm}
    \end{equation}
\end{lemma}

\begin{proof}
    Suppose that~\eqref{eq:geometry_2} holds.  If $\|\a\str(t)\|\leq\tfrac{r}{2}$, then $\calC(t)\cap\calD(t)=\calC(t)$.  Hence, $\|\a\str(t)\| [\cos\theta(t) \, \sin\theta(t)]\tt\in\calC(t)=\calC(t)\cap\calD(t)\subset\calD(t)$ implies~\eqref{eq:set_membership_objective}.  If $\|\a\str(t)\|>\tfrac{r}{2}$, \vspace{-0.12cm}
    \begin{equation}
        \hspace{-0.08cm}
        \|\a\str(t)\| \!\! \bbm{\cos\theta(t) \\ \sin\theta(t)} \!\! \in \! \|\a\str(t)\| \!\! \bbm{\cos\bar\Theta(t) \\ \sin\bar\Theta(t)} \!\! 
        \subset \calD(t), \vspace{-0.05cm}
    \end{equation}
    so $\a(t) =\|\a\str(t)\| \lambda(t)\in\calA(t)$ after rotating back to the original coordinates.  

    Suppose that~\eqref{eq:set_membership_objective} holds.  If $\|\a\str(t)\|\leq\tfrac{r}{2}$, then $\bar\Theta(t)=\bar\RR$ so $\theta(t)\in\bar\Theta(t)$ trivially.  If $\|\a\str(t)\|>\tfrac{r}{2}$, then $\bar\Theta(t)$ is the set of angles corresponding to $\calC(t)\cap\calD(t)$.   From~\eqref{eq:geometry_1},~\eqref{eq:set_membership_objective} is equivalent to $\|\a\str(t)\| [\cos\theta(t) \, \sin\theta(t)]\in\calD(t)$.  In addition, $\|\a\str(t)\| [\cos\theta(t) \, \sin\theta(t)]\in\calC(t)$ from~\eqref{eq:defn_of_C}.  Therefore, \vspace{-0.17cm}
    \begin{equation}
        \hspace{-0.1cm}
        \|\a\str(t)\| \! \bbm{\cos\theta(t) \\ \sin\theta(t)} \in \calC(t)\cap\calD(t) = \|\a\str(t)\| \! \bbm{\cos\bar\Theta(t) \\ \sin\bar\Theta(t)}\!.
        \label{eq:lands_in_sliver} \vspace{-0.17cm}
    \end{equation}
    so $\theta(t) + 2\pi k\in\bar\Theta(t)$ for some $k\in\mathbb{Z}$.  Recall that the definition~\eqref{eq:defn_of_theta} ensures $\theta(t)\in[-\pi,\pi)$.  Pairing this fact with~\eqref{eq:lands_in_sliver} results in~\eqref{eq:geometry_2}.  
\end{proof}

While~\eqref{eq:geometry_2} is sufficient (and necessary) for~\eqref{eq:set_membership_objective}, the angle $\bar\theta(\|\a\str(t)\|/r)$
does not vary smoothly with $\|\a\str(t)\|$ and so cannot be used in a backstepping design.  Towards enabling the use of backstepping, a smooth underestimate $\theta\str(\cdot)$ of $\bar\theta(\cdot)$ is introduced.  Referring to Figure~\ref{fig:angular_bounds}, note that $\theta\str(\sigma)\leq\bar\theta(\sigma)$ for all $\sigma\in\bar\RR$ if and only if $\theta\str(s^{-1})\leq\bar\theta(s^{-1})$ for all $s\in\bar\RR$.  It is advantageous to design $s\mapsto\theta\str(s^{-1})$ rather than $\sigma\mapsto\theta\str(\sigma)$ since this choice enables the use of the identity $\arccos(1-\frac{s^2}{2})\geq|s|$ for all $s\in[-1,1]$.  It is also advantageous to design $\theta\str(\cdot)$ so that $\theta\str(\sigma)=2\pi$ for all $\sigma$ in a neighborhood of zero.  Hence, $\theta\str(\cdot)$ is selected to be \vspace{-0.17cm}
\begin{equation}
    \theta\str(s^{-1}) = |s| \bigl(1-\zeta(|s|;2,2\pi)\bigr) + 2\pi\,\zeta(|s|;2,2\pi)
    \label{eq:wacky_defn_of_theta_star}
    \vspace{-0.17cm}
\end{equation}
where $\zeta$ is a smooth step function defined in Appendix~\ref{appsec:zeta}. The function $\zeta(|s|;2,2\pi)$ rises from zero at $|s|=2$ to one at $|s|=2\pi$.  Equation~\eqref{eq:wacky_defn_of_theta_star} can be equivalently rewritten as \vspace{-0.17cm}
\begin{equation}
    \theta\str(\sigma) = |\sigma|^{-1} \! \bigl(1\!-\!\zeta(|\sigma|^{-1};2,2\pi)\bigr) + 2\pi\,\zeta(|\sigma|^{-1};2,2\pi) \label{eq:defn_of_theta_star}
    \vspace{-0.17cm}
\end{equation}
when $\sigma\neq0$ with $\theta\str(0)=2\pi$.  By construction, $\theta\str(\cdot)$ is smooth and $\theta\str(\sigma) \leq \bar\theta(\sigma)$ for all $\sigma\in\RR$.  This underestimate provides a corollary to Lemma~\ref{lemma:theta_bar} which is the desired sufficient condition for~\eqref{eq:set_membership_objective}.

\begin{corollary} \label{corollary:theta_star}
    Define $\Theta\str(\cdot)$ by \vspace{-0.17cm}
    \begin{equation}
        \label{eq:defn_of_Theta_star}
        \Theta\str(t) = [-\theta\str(\|\a\str(t)\|/r), \theta\str(\|\a\str(t)\|/r)]. \vspace{-0.17cm}
    \end{equation}
    Then, \vspace{-0.17cm}
    \begin{equation}
        \theta(t)\in\Theta\str(t) \quad \Rightarrow \quad \a(t)=\lambda(t)u_1(t)\in\calA(t).
        \vspace{0.1cm}
    \end{equation}
\end{corollary}

\begin{proof}
    Note $\theta\str(\sigma) \leq \bar\theta(\sigma)$ for $\sigma\in\RR$ ensures $\Theta\str(t)\subseteq\bar\Theta(t)$.  From Lemma~\ref{lemma:theta_bar}, $\theta(t)\in\Theta\str(t)\subseteq\bar\Theta(t)$ implies~\eqref{eq:set_membership_objective}.  
\end{proof}

\begin{figure}[t!]
    \centering
    \begin{subfigure}{0.5\textwidth}
        \centering
        \includegraphics[trim={0.0cm 0cm 0.0cm -0.4cm},clip,width=0.95\textwidth]{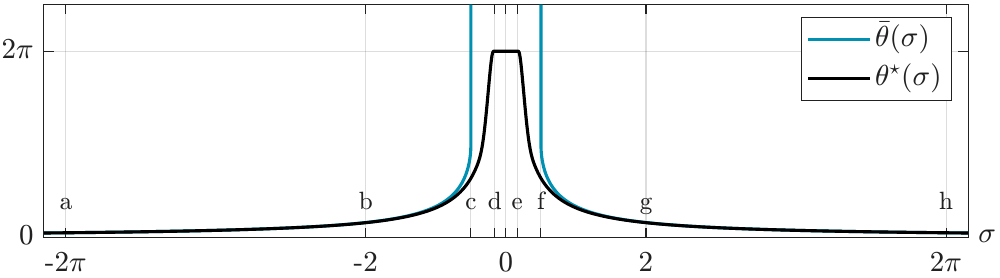}
        \caption{Graphs of $\bar\theta(\sigma)$ and $\theta\str(\sigma)$.}
        \label{fig:thetas_with_identity_arg}
    \end{subfigure}\hfill 
    \begin{subfigure}{0.5\textwidth}
        \centering
        \includegraphics[width=0.95\textwidth]{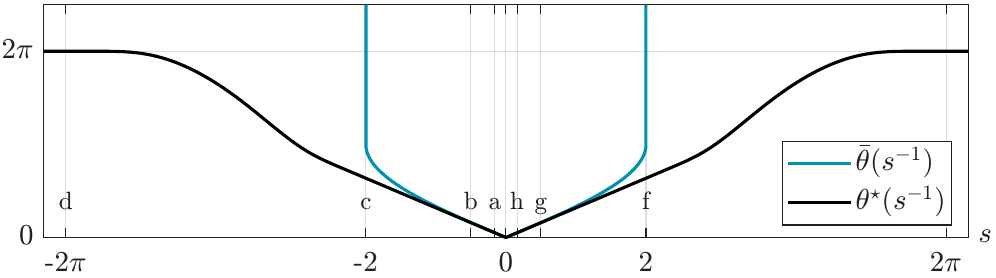}
        \caption{Graphs of $\bar\theta(s^{-1})$ and $\theta\str(s^{-1})$.}
        \label{fig:thetas_with_inverse_arg}
    \end{subfigure}\hfill 
    \caption{Angular bounds.  Correspondence of horizontal stations between (a) and (b) is indicated with matching letters.}
    \label{fig:angular_bounds}
    \vspace{-0.4cm}
\end{figure}

\subsubsection{Introducing 1D MRP Coordinates} \label{subsec:introducing_1D_MRPs!}
To resolve the topological obstruction at $\lambda(t)=-\lambda\str(t)$, an alternative representation of the pivot's attitude is introduced.  Define $\eta:\RR_{\geq0}\to\bar\RR$ by continuity and the inclusion \vspace{-0.17cm}
\begin{equation}
    \eta(t)\in\{\tan(\tfrac{1}{4}\theta(t)),-\cot(\tfrac{1}{4}\theta(t))\}
    \label{eq:defn_of_eta} \vspace{-0.17cm}
\end{equation}
with $\eta(0) = \tan(\tfrac{1}{4}\theta(0))$.  The case where $\theta(0)$ is not defined is addressed by a switching rule presented in Section~\ref{subsec:switching_condition}.  

Using the coordinate $\eta(t)$ to describe the orientation of the pivot is analogous to using MRPs to describe orientation in $SO(3)$.  Accordingly, the $\eta$ coordinates are hereafter referred to as 1D MRP coordinates.  The set of 1D MRPs is $\bar\RR$.  Infinities are included in this set since they are well defined 1D MRPs that correspond to the identity rotation.  The coordinate $\eta$ provides a double covering of $SO(2)$. Rotations that are $180^{\circ}$ away from $\theta(t)=0$ are mapped to the unit circle (which is $\{-1,1\}$ for $\RR^1$).  

Using 1D MRPs resolves the topological obstruction at $\lambda(t)=-\lambda\str(t)$.  If $\lambda(t)=-\lambda\str(t)$ is mapped to $\eta(t)=-1$ then turning counterclockwise drives $\eta(t)\to0$, and if $\lambda(t)=-\lambda\str(t)$ is mapped to $\eta(t)=1$ then turning clockwise drives $\eta(t)\to0$.  Sometimes, it is appropriate to switch the 1D MRP representing the pivot orientation, as explained in Section~\ref{subsec:switching_condition}.  When it is not switching, $\eta(t)$ evolves continuously according to \vspace{-0.17cm}
\begin{equation}
    \dot\eta(t) = \frac{1+\eta(t)^2}{4}\left(\omega(t)-(\calS\lambda\str(t))\tt \frac{\dot\a\str(t)}{\|\a\str(t)\|}\right),
    \label{eq:defn_of_ddt_eta} \vspace{-0.17cm}
\end{equation}
which follows from taking the derivative of~\eqref{eq:defn_of_eta}.

Each 1D MRP can be mapped to a \textit{shadow MRP} that corresponds to the same rotation.  The function that converts between MRPs and shadow MRPs is ${\rm sh}(\eta) = -\eta^{-1}$.  

The 1D MRP counterpart to $\Theta\str(t)$ in~\eqref{eq:defn_of_Theta_star} is the set \vspace{-0.17cm}
\begin{equation}
    N\str(t) = \left[-\tan(\tfrac{1}{4}\theta\str(\tfrac{\|\a\str(t)\|}{r})),\tan(\tfrac{1}{4}\theta\str(\tfrac{\|\a\str(t)\|}{r}))\right].
    \vspace{-0.17cm}
\end{equation}
For later convenience, define \vspace{-0.17cm}
\begin{equation}
    \eta\str(t)=\tan(\tfrac{1}{4}\theta\str(\|\a\str(t)\|/r)),
    \vspace{-0.17cm}
\end{equation}
so that $N\str(t) = [-\eta\str(t),\eta\str(t)]$.

Note that $\eta(t)$ depends on $\theta(t)$ in~\eqref{eq:defn_of_eta} so it becomes undefined if $\a\str(t)=0$.  Accordingly, let $\eta(t)$ be defined only on all time intervals where $\a\str(t)\neq0$ and assume that each of these intervals has nonzero measure.  The fact that $\eta(t)$ can become undefined is not an issue for the upcoming Lyapunov analysis since the Lyapunov function proposed in the next section is designed to be independent of $\eta(t)$ during any such periods.  (This is achieved by expanding $N\str(t)$ to $\bar\RR$ whenever $\a\str(t)=0$.)

\subsubsection{Designing a Lyapunov Function}

To guarantee convergence of $\a(t)=\lambda(t)u_1(t)$ to $\setA(t)$ through smooth feedback, we first construct a Lyapunov function for use in a backstepping control design process of Section~\ref{subsubsec:backstepping}.  MRP-based attitude controllers can, in principle, exhibit unwinding.  This unwinding makes global asymptotic stability (GAS) of $\setA(t)$ impossible --- stability in the sense of Lyapunov fails since it is possible for $\a(t)$ to begin close to $\setA(t)$, move away from the ball as $\eta(t)$ unwinds, and then return as $t\to\infty$.  Instead of proving GAS, the set $\setA(t)$ is shown to be globally attractive.  The proof relies on the development of a candidate Lyapunov function $V_\a(t,\eta)$ for which \vspace{-0.17cm}
\begin{equation}
    \|\a(t)\|_{\setA(t)}\leq \mu\bigl(V_\a(t,\eta(t))\bigr)
    \label{eq:overestimate_goal}
    \vspace{-0.17cm}
\end{equation}
for some $\mu\in\calK_{\infty}$.  

Suppose $\a(t)\not\in\setA(t)$ as depicted in Figure~\ref{fig:arclength_construction}.  Let $\alpha(t)$ denote the point lying at the intersection of $\partial\setA(t)$ with the line connecting $\a(t)$ and $\a\str(t)$. Then $\|\a(t)\|_{\setA(t)} = \|\a(t)-\alpha(t)\|$.  Consider the two points that lie at the intersection of $\partial\setA(t)$ and $\|\a\str(t)\|\partial\BB^2$.  Let $\beta(t)$ denote one of these points. If one is closer to $\a(t)$, take $\beta(t)$ to be that point. Note that $\|\a(t)-\alpha(t)\|\leq\|\a(t)-\beta(t)\|$ since $\alpha(t)$ is the point closest to $\a(t)$ on the boundary of $\setA(t)$.  It is also true that $\|\a(t)-\beta(t)\|$ is less than the length of the minor arc connecting $\beta(t)$ to $\a(t)$ along $\|\a\str(t)\|\partial\BB^2$.  The length of this arc is $\|\a\str(t)\| (|\theta(t)| - \bar\theta(\tfrac{\|\a\str(t)\|}{r}))$.  With the reasoning above, \vspace{-0.17cm}
\begin{equation}
    \|\a(t)\|_{\setA(t)} \leq \|\a\str(t)\| (|\theta(t)| - \theta\str(\tfrac{\|\a\str(t)\|}{r})), \label{eq:intermediate_overestimate}
    \vspace{-0.17cm}
\end{equation}
where $\bar\theta$ has been replaced with its smooth underestimate $\theta\str$.  

\begin{figure}[t!]
    \begin{minipage}{0.5\textwidth}
        \centering
        \includegraphics[trim={0.0cm 1.5cm 0.0cm 1.5cm},clip,width=0.6\textwidth]{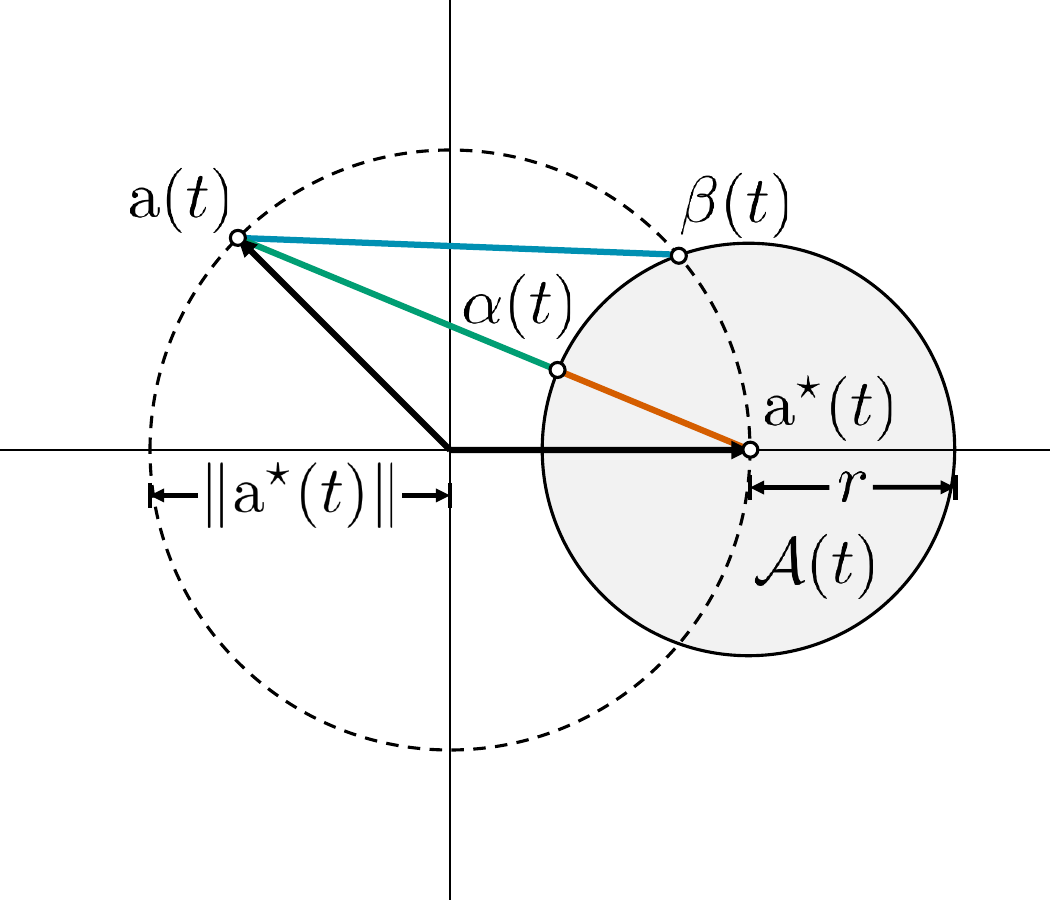}
        \caption{Schematic illustrating $\a(t)\not\in\setA(t)$.\vspace{-0.17cm}}
        \label{fig:arclength_construction}
    \end{minipage}\hfill \vspace{-0.5cm}
\end{figure}

The signed arclength around a unit circle corresponding to a 1D MRP $\eta$ is $4\arctan(\eta)$, which takes values in $[-2\pi,2\pi]$ since the 1D MRPs are a double covering of $\SS^1$.  This expression can be used to rewrite~\eqref{eq:intermediate_overestimate} as \vspace{-0.17cm}
\begin{equation}
    \hspace{-0.17cm}\|\a(t)\|_{\setA(t)} \!\leq\! 4\|\a\str(t)\| \bigl(|\arctan(\eta(t))|\!-\!\arctan(\eta\str(t))\bigr).\!\!  \label{eq:radially_bounded_V_a} \vspace{-0.17cm}
\end{equation}
To ensure $V_{\a}$ is radially unbounded with respect to $\eta$, Lemma~\ref{lemma:arctan_stuff} from the appendix is applied to obtain \vspace{-0.17cm}
\begin{equation}
    \|\a(t)\|_{\setA(t)}\leq 4\|\a\str(t)\| \bigl(|\eta(t)|-\eta\str(t)\bigr). \label{eq:radially_unbounded_V_a} \vspace{-0.17cm}
\end{equation}
Using the fact that $\chi\in\calK_{\infty}\Leftrightarrow\chi^{-1}\in\calK_{\infty}$, the function $\mu\in\calK_{\infty}$ is now defined by its inverse \vspace{-0.17cm}
\begin{equation}
    \mu^{-1}(s) = Z(\tfrac{1}{4}s;0,\Delta_\a) \vspace{-0.17cm}
\end{equation}
for some $\Delta_\a>0$, where $Z$ is a smooth ramp function defined in Appendix~\ref{appsec:zeta} that is parameterized similarly to the smooth step function $\zeta$ introduced earlier. Then, \vspace{-0.17cm}
\begin{align}
    \hspace{-0.17cm}
    \|\a(t)\|_{\setA(t)} &\leq  \mu(\mu^{-1}(4\|\a\str(t)\| \bigl(|\eta(t)|\!-\!\eta\str(t)\bigr))) 
    \\
    \hspace{-0.17cm}
    \Leftrightarrow \|\a(t)\|_{\setA(t)} &\leq \mu\bigl(Z(\|\a\str(t)\| (|\eta(t)|\!-\!\eta\str(t));0,\Delta_\a)\bigr),\\[-0.65cm] \notag
\end{align}
which is the desired bound~\eqref{eq:overestimate_goal} when $V_\a(t,\eta)$ is defined by\vspace{-0.17cm}
\begin{equation}
    V_\a(t,\eta) = Z(\|\a\str(t)\| (|\eta|-\eta\str(t));0,\Delta_\a). \label{eq:defn_of_V_eta} \vspace{-0.17cm}
\end{equation}
With this definition, we are ready to apply backstepping.

\subsubsection{Backstepping Analysis} \label{subsubsec:backstepping}
Consider the candidate Lyapunov function $V_\a(t,\eta)$.  Omitting arguments of time and simplifying, the derivative of $V_\a$ along trajectories is \vspace{-0.17cm}
\begin{equation}
    \begin{aligned}
        &\dot V_\a
        = \zeta(\|\a\str\|(|\eta|-\eta\str);0,\Delta_\a) \Bigl(\lambda\str{}\tt\dot\a\str (|\eta|-\eta\str)
        \\& 
        + \|\a\str\|\Bigl(\signof(\eta)\bigl(\tfrac{1+\eta^2}{4}\bigr)\bigl(\omega -(\calS\lambda\str)\tt \tfrac{\dot\a\str}{\|\a\str\|}\bigr) -\dot \eta\str \Bigr) \Bigr).
    \end{aligned} \label{eq:ddt_V_eta_1} \vspace{-0.17cm}
\end{equation}

Define the target angular velocity $\omega\str:\RR_{\geq0}\times\RR\to\RR$ by \vspace{-0.17cm}
\begin{align}
    \omega\str(t,\eta) 
    &= 
    \biggl(-\lambda\str(t)\tt\Bigl(\frac{4\signof(\eta)}{1+\eta^2}(|\eta|\!-\!\eta\str(t))I \!+\! \calS \Bigr)\frac{\dot\a\str(t)}{\|\a\str(t)\|}   
    \notag
    \\&\qquad
    +\frac{4\signof(\eta)}{1+\eta^2}\Bigl(\dot\eta\str(t)\bigl(1-\zeta(\dot \eta\str(t);0,\Delta_{\dot\eta\str})\bigr)\Bigr)
    \notag
    \\&\qquad
    -\frac{4\signof(\eta)}{(1+\eta^2)}\frac{\Omega_\a(t,\eta)}{\|\a\str(t)\|}\biggr) \zeta(\frac{|\eta(t)|}{\eta\str(t)};\varrho,1)
    \notag
    \\&\qquad
    - \frac{4 k_{\eta}\eta}{1+\eta^2} \zeta(\|\a\str(t)\|;\a_0,\a_1)
    \label{eq:defn_of_omega_star} \vspace{-0.17cm}
\end{align} 
whenever $\a\str(t)\neq0$ with $\omega\str(t,\eta)=0$ otherwise, for some $\Delta_{\dot\eta\str}\in(0,\infty)$, $\varrho\in(0,1)$, $k_{\eta}\in(0,\infty)$, $\a_0\in(0,\infty)$, and $\a_1\in(\a_0,\infty)$, where $\Omega_\a:\RR_{\geq0}\times\bar\RR\to\RR$ is yet to be specified.  The first term on the right of~\eqref{eq:defn_of_omega_star}, the one involving $\dot\a\str(t)$, is a feedforward term responsible for driving the pivot attitude $\lambda(t)$ toward the desired attitude $\lambda\str(t)$ and for speeding up the convergence of the rotational coordinate $\eta(t)$ to the set of target attitudes $N\str(t)$ whenever the magnitude of $\a\str(t)$ is growing rapidly.  The second term in~\eqref{eq:defn_of_omega_star}, the one involving $\dot\eta\str(t)$, is responsible for keeping $\eta(t)$ inside of $N\str(t)$.  The third term in~\eqref{eq:defn_of_omega_star}, the one involving $\Omega_\a$, is the one primarily responsible for driving $V_\a$ to zero in the subsequent Lyapunov analysis.  The fourth and final term in~\eqref{eq:defn_of_omega_star} keeps $\eta(t)$ small once it reaches $N\str(t)$.  

In conventional backstepping designs for multirotors~\cite{lee_geometric_2010,lee_control_2011,lee_geometric_2011,lee_nonlinear_2013,gamagedara_geometric_2019,martins_global_2021,martins_global_2024,flores_robust_2025}, the angular velocity feedforward terms present a singularity at $\a\str=0$.  The proposed design does not suffer this issue.  Suppose that $\a\str\to0$.  Then before $\a\str=0$, the target acceleration must pass through the neighborhood of $0$ where $\|\a\str\|\in(0,\tfrac{r}{2\pi}]$.  Within this interval, $N\str(t)=\RR$ and $\eta\str(t)=\infty$.  It follows that $|\eta(t)|/\eta\str(t)\leq\varrho$ so $\zeta(|\eta(t)|/\eta\str(t);\varrho,1)=0$.  Hence, the first three terms in~\eqref{eq:defn_of_omega_star} smoothly transition to zero whenever $\|\a\str\|$ enters $[0,\tfrac{r}{2\pi}]$.  Colloquially, these three control actions are \textit{turned off} whenever $\|\a\str\|\leq\tfrac{r}{2\pi}$ which ensures that the free-fall singularity does not cause $\omega\str(t,\eta(t))$ to become discontinuous along any trajectory.  

The Lyapunov function is specifically designed to accommodate the three terms being smoothly switched on and off.  The switching has no direct impact on $\dot{V}_{\a}$ as \vspace{-0.17cm}
\begin{equation}
    \begin{aligned}
        \zeta(|\eta|/\eta\str&;\varrho,1) \, \zeta(\|\a\str\|(|\eta| \!-\! \eta\str);0,\Delta_\a) \\&= \zeta(\|\a\str\|(|\eta| \!-\! \eta\str);0,\Delta_\a) 
    \end{aligned} \label{eq:vanishing_step} \vspace{-0.17cm}
\end{equation}
since $\zeta(|\eta|/\eta\str;\varrho,1)=1$ whenever $\zeta(\|\a\str\|(|\eta| \!-\! \eta\str);0,\Delta_\a)$ is nonzero.  The identity~\eqref{eq:vanishing_step} enables~\eqref{eq:defn_of_omega_star} to avoid singularities at $\a\str(t)=0$ while ensuring the following desirable evolution of ${V}_{\a}$. Adding and subtracting $\omega\str$ on the right hand side of~\eqref{eq:ddt_V_eta_1} then simplifying reveals \vspace{-0.17cm}
\begin{align}
    \dot V_\a &= - \Bigl(\Omega_\a + \dot \eta\str\,\|\a\str\|\,\zeta(\dot \eta\str;0,\Delta_{\dot\eta\str}) 
    \label{eq:ddt_V_eta_2}
    \\ & \quad 
    + k_\eta |\eta| \|\a\str\| \, \zeta(\|\a\str\|;\a_0,\a_1)  \Bigr)  \zeta(\|\a\str\|(|\eta|\!-\!\eta\str);0,\Delta_\a) 
    \notag
    \\ & 
    + \signof(\eta)\bigl(\tfrac{1+\eta^2}{4}\bigr) \|\a\str\| \zeta(\|\a\str\|(|\eta|\!-\!\eta\str);0,\Delta_\a) 
    \bigl( \omega - \omega\str \bigr). \notag \\[-0.65cm] \notag
\end{align} 
To ensure $V_\a\to0$, take $k_\a >0$ and prescribe \vspace{-0.17cm}
\begin{equation}
    \hspace{-0.17cm}
    \Omega_\a(t,\eta) = \left\{ 
    \matl{
    k_\a  \frac{Z(\|\a\str(t)\|(|\eta| - \eta\str(t));0,\Delta_\a)}{\zeta(\|\a\str(t)\|(|\eta| - \eta\str(t));0,\Delta_\a)}, & \!\!\!\!V_\a(t,\eta)>0 \\
    0, & \!\!\!\!\text{otherwise}
    }.
    \right. \hspace{-0.1cm} \label{eq:defn_of_Omega_eta} \vspace{-0.17cm}
\end{equation}
which is smooth and nonnegative.  With this choice, the term involving $\Omega_\a$ in~\eqref{eq:ddt_V_eta_2} becomes \vspace{-0.17cm}
\begin{equation}
    -\Omega_\a \, \zeta(\|\a\str\|(|\eta|\!-\!\eta\str);0,\Delta_\a)  = -k_\a  V_\a. \vspace{-0.17cm}
\end{equation}

Expanding the scope of the backstepping design to the entire inner loop, consider the candidate Lyapunov function \vspace{-0.17cm}
\begin{equation}
    V(t,\eta,\omega) = V_\a(t,\eta) + \underbrace{\frac{1}{2} p_{\omega} (\omega-\omega\str(t,\eta))^2}_{\eqqcolon \, V_\omega(t,\eta,\omega)}, \vspace{-0.17cm}
\end{equation}
where $p_\omega>0$.  Omitting arguments of time, the derivative of $V$ along trajectories is \vspace{-0.17cm}
\begin{align}
    \dot V 
    &= \dot V_\a + p_{\omega} (\omega-\omega\str)(u_2 - \dot\omega\str)
    \label{eq:ddt_V_1} \\ 
    &= -k_\a V_\a - \Bigl(\dot \eta\str\|\a\str\|\zeta(\dot \eta\str;0,\Delta_{\dot\eta\str}) 
    \notag
    \\ & \qquad 
    + k_\eta |\eta| \|\a\str\| \, \zeta(\|\a\str\|;\a_0,\a_1)  \Bigr)  \zeta(\|\a\str\|(|\eta|\!-\!\eta\str);0,\Delta_\a) 
    \notag \\ & \quad 
    + p_{\omega} (\omega-\omega\str)\Bigr(u_2 - \dot\omega\str
    \label{eq:ddt_V_3} 
    \\ & \qquad 
    + p_{\omega}^{-1}\signof(\eta)\bigl(\tfrac{1+\eta^2}{4}\bigr) \|\a\str\| \zeta(\|\a\str\|(|\eta|\!-\!\eta\str);0,\Delta_\a) \Bigr).
    \notag \\[-0.65cm] \notag
\end{align}
Set $k_\omega >0$ and prescribe \vspace{-0.17cm}
\begin{equation}
    \begin{aligned}
        u_2(t) &= \dot\omega\str(t,\eta(t),\omega(t)) - \tfrac{1}{2} k_{\omega} (\omega(t)-\omega\str(t,\eta(t))) 
        \\ & \quad 
        -p_{\omega}^{-1}\signof(\eta(t))\bigl(\tfrac{1+\eta(t)^2}{4}\bigr)\|\a\str(t)\| 
        \\ & \qquad 
        \cdot \, \zeta(\|\a\str(t)\|(|\eta(t)|\!-\!\eta\str(t));0,\Delta_\a).
    \end{aligned} \label{eq:u2_prescription} \vspace{-0.17cm}
\end{equation}
Combining~\eqref{eq:ddt_V_3} with~\eqref{eq:u2_prescription} and simplifying reveals \vspace{-0.17cm}
\begin{equation}
    \dot V \leq -k_\a V_\a - k_{\omega}V_\omega \leq -k V, \label{eq:sufficient_for_attractivity} \vspace{-0.17cm}
\end{equation}
where $k=\min\{k_\a ,k_{\omega}\}$.  Applying the Comparison Lemma to this result yields the following theorem.  

\begin{theorem} \label{thrm:a_to_A}
    Under the action of the control law defined in~\eqref{eq:u1_prescription} and~\eqref{eq:u2_prescription}, $\setA(t)$ globally attracts $\a(t)=\lambda(t)u_1(t)$.  
\end{theorem}

\begin{proof}
    Equation~\eqref{eq:sufficient_for_attractivity} and the Comparison Lemma~\cite{khalil_nonlinear_2002} together imply that $V$ decays exponentially to zero: \vspace{-0.17cm}
    \begin{equation}
        V(t,\eta(t),\omega(t))\leq V(T,\eta(T),\omega(T)) e^{-k(t-T)}
        \label{eq:decaying_bound_on_V} \vspace{-0.17cm}
    \end{equation}
    for all $T\geq0$.  Recalling~\eqref{eq:overestimate_goal} and that $V_\a\leq V$, \vspace{-0.17cm}
    \begin{equation}
        \|\a(t)\|_{\setA(t)}\leq \mu\bigl(V(T,\eta(T),\omega(T)) e^{-k(t-T)}\bigr) 
        \label{eq:decaying_bound_on_set_dist} \vspace{-0.17cm}
    \end{equation}
    for all $T\geq0$.  Selecting $T=0$ and taking $t\to\infty$ shows that $\|\a(t)\|_{\setA(t)}\to0$, so $\a(t)\to\setA(t)$ globally.
\end{proof}

\subsubsection{MRP Switching Condition} \label{subsec:switching_condition}

Before analyzing the behavior of the outer-loop system in light of the established results, this section introduces the MRP switching condition for the control law.  Formal notation for hybrid systems~\cite{goebel_hybrid_2009} is omitted for brevity. 

To accommodate the possibility of $\eta(t)$ becoming undefined at $\a\str(t)=0$, the control law is designed to become independent of $\eta(t)$ for $\|\a\str(t)\|<\min\{\tfrac{r}{2\pi},\a_0\}$.  Suppose $\a\str(t)$ crosses through zero.  While $\|\a\str(t)\|<\min\{\tfrac{r}{2\pi},\a_0\}$, the behavior of the control law is unaffected.  The value of $\eta(t)$ after the crossing is, however, ambiguous since the 1D MRPs are a double covering of $\SS^1$.  To resolve this ambiguity, we apply the switching condition: \vspace{-0.17cm}
\begin{equation}
    \begin{aligned}
        \|\a\str(t)\| \geq \min\{\tfrac{r}{2\pi},\a_0\} \,\, &\Rightarrow \,\, \text{$\dot\eta(t)$ given by~\eqref{eq:defn_of_ddt_eta}}
        \\
        \|\a\str(t)\|\in(0,\min\{\tfrac{r}{2\pi},\a_0\}) \,\, &\Rightarrow \,\, \eta(t) = \tan(\tfrac{1}{4}\theta(t)). 
    \end{aligned}
    \label{eq:switching_condition} \vspace{-0.17cm}
\end{equation}
This switching condition does not cause any discontinuous control action since it occurs only when the control law has no dependence on $\eta(t)$.  It has the added benefit of discouraging unwinding since the $\tan(\tfrac{1}{4}\theta(t))$ branch of~\eqref{eq:defn_of_eta} is the one closer to the origin.

\subsection{Long-term Behavior of the Composite System} \label{subsec:long_term}

Consider the behavior of the inner- and outer-loop systems when they are connected to one another.  Set $\a\str(t) = \kappa_{ x\str}(t,e(t))$ and define $\Delta\a:\RR_{\geq0}\to\RR^2$ according to \vspace{-0.17cm}
\begin{equation}
    \Delta\a(t)= \a(t)-\a\str(t). \vspace{-0.17cm}
\end{equation}
It follows that the first row of~\eqref{eq:olsys} can be rewritten in terms of $e(t) =  x(t)- x\str(t)$ as \vspace{-0.17cm}
\begin{equation}
    \begin{aligned}
        \dot e(t) &= f_{ x\str}(t,e(t)) + \g_{ x\str}(t,e(t))\kappa_{ x\str}(t,e(t)) 
        \\
        &\qquad 
        + \g_{ x\str}(t,e(t)) \Delta\a(t).
    \end{aligned} \label{eq:error_outer_clsys} \vspace{-0.12cm}
\end{equation}
The form of~\eqref{eq:error_outer_clsys} matches~\eqref{eq:error_surrogate_olsys} with $u_\a(t) = \kappa_{ x\str}(t,e(t))$ and $\w(t) = \g_{ x\str}(t,e(t)) \Delta\a(t)$.  The ISS property~\eqref{eq:outer_clsys_is_pre_ISS} of the outer-loop system~\eqref{eq:surrogate_olsys}, therefore, implies that for any $T\in[0,\infty)$,\vspace{-0.17cm}
\begin{equation}
    \hspace{-0.17cm}
    \|e(t)\| \leq \beta(\|e(T)\|,t\!-\!T) + \gamma({\sup}_{\tau\geq T} \, \bar\g \|\Delta\a(\tau)\|) \hspace{-0.1cm} \label{eq:using_ISS} \vspace{-0.17cm}
\end{equation}
for all $t\in [T,\infty)$, where $\bar\g=\sup_{t\in\RR_{\geq0},x\in\RR^n}\|\g(t, x)\|$. Theorem~\ref{thrm:a_to_A} provides a bound on $\|\Delta\a(t)\|\leq r+\|\a(t)\|_{\setA(t)}$.  This bound is applied to strengthen the ISS-type estimate~\eqref{eq:using_ISS} in the two following theorems.  

\begin{theorem} \label{thrm:asymptotic_stabilization}
    Under~\eqref{eq:u1_prescription},~\eqref{eq:u2_prescription}, for any $T\in[0,\infty)$ the ball\vspace{-0.17cm}
    \begin{equation}
        \calB_T=\gamma\bigl(\bar\g r + \bar\g \mu(V(T,\eta(T),\omega(T)))\bigr)\bar\BB^n \vspace{-0.17cm}
    \end{equation}
    is global asymptotic stability (GAS) on time interval $[T,\infty)$.
\end{theorem}

\begin{proof}
    Consider~\eqref{eq:using_ISS}.  From $\|\Delta\a(\tau)\|\leq r+\|\a(\tau)\|_{\setA(\tau)}$ and~\eqref{eq:decaying_bound_on_set_dist}, the second term on the right is bounded by \vspace{-0.17cm}
    \begin{align}
        \hspace{-0.17cm}
        \gamma({\sup}_{\tau\geq T} \, \bar\g \|\Delta\a(\tau)\|) 
        &\leq 
        \gamma(\bar\g r \!+\! {\sup}_{\tau\geq T} \, \bar\g \mu(V_T e^{-k(\tau\!-\!T)}))
        \hspace{-0.17cm}
        \notag
        \\&=
        \gamma(\bar\g r + \bar\g \mu(V_T)) \\[-0.65cm] \notag
    \end{align}
    for any $T\in[0,\infty)$, where $V_T=V(T,\eta(T),\omega(T))$.  
    Hence,\vspace{-0.17cm}
    \begin{equation}
        \|e(t)\| \leq \beta(\|e(T)\|,t\!-\!T) +\gamma(\bar\g r + \bar\g \mu(V_T)) \vspace{-0.17cm}
        \label{eq:did_not_expect_to_reference_this_TWICE_haha}
    \end{equation}
    from~\eqref{eq:using_ISS}, for all $t\in[T,\infty)$, for any $T\in[0,\infty)$.  Note that  \vspace{-0.17cm}
    \begin{equation}
        \|e(t)\|_{\calB_T}=\max\{0,\|e(t)\|-\gamma(\bar\g r + \bar\g \mu(V_T))\}. \vspace{-0.17cm}
    \end{equation}
    Then subtracting $\gamma(\bar\g r + \bar\g \mu(V_T))$ from both sides of~\eqref{eq:did_not_expect_to_reference_this_TWICE_haha} and using the fact that $0\leq\beta(\|e(T)\|,t\!-\!T)$, it follows that \vspace{-0.17cm}
    \begin{equation}
        \|e(t)\|_{\calB_T} \leq \beta(\|e(T)\|,t\!-\!T) \label{eq:yet_another_reference} \vspace{-0.17cm}
    \end{equation}
    for all $t\in[T,\infty)$, for any $T\in[0,\infty)$. The estimate~\eqref{eq:yet_another_reference} proves GAS of $\calB_T$ on the time interval $[T,\infty)$~\cite{chaillet_uniform_2008}.  
\end{proof}

\begin{theorem} \label{thrm:asymptotic_tracking}
    Under~\eqref{eq:u1_prescription},~\eqref{eq:u2_prescription}, the outer-loop error is ultimately bounded according to $\limsup{}_{t\to\infty} \|e(t)\| \leq \gamma(\bar\g r)$. 
\end{theorem}

\begin{proof}
    Consider~\eqref{eq:did_not_expect_to_reference_this_TWICE_haha}.  Taking the limit supremum as $t\to\infty$ drives the $\calK\calL$ term to zero, leaving \vspace{-0.17cm}
    \begin{equation}
        \hspace{-0.17cm}
        \limsup_{t\to\infty} \|e(t)\| 
        \leq \gamma(\bar\g r + \bar\g \mu(V_T))  \quad \forall \, T\in[0,\infty). \vspace{-0.17cm}
    \end{equation}
    Let $V_0=V(0,\eta(0),\omega(0))$.  Using $V_T\leq V_0 e^{-kT}$ from~\eqref{eq:decaying_bound_on_V}, \vspace{-0.1cm}
    \begin{equation}
        \hspace{-0.15cm}
        \limsup_{t\to\infty} \|e(t)\| 
        \leq \gamma\bigl(\bar\g r+\bar\g\mu(V_0 e^{-kT}) \bigr) \quad \forall \, T\in[0,\infty).
        \label{eq:ultimate_bound_with_T} \vspace{-0.17cm}
    \end{equation}
    Since the right hand side of~\eqref{eq:ultimate_bound_with_T} is continuous in $T$ and it holds for all $T\in[0,\infty)$,~\eqref{eq:ultimate_bound_with_T} must hold as $T\to\infty$. Hence,\vspace{-0.1cm}
    \begin{equation}
        \hspace{-0.15cm}
        \limsup_{t\to\infty} \|e(t)\| 
        \leq \lim_{T\to\infty}\gamma\bigl(\bar\g r+\bar\g \mu(V_0 e^{-kT}) \bigr) = \gamma(\bar\g r)\! \vspace{-0.17cm}
    \end{equation}
    which is the claimed result.
\end{proof}

Theorems~\ref{thrm:asymptotic_stabilization} and~\ref{thrm:asymptotic_tracking} describe the transient and long-term behavior of the outer-loop system.  Theorem~\ref{thrm:asymptotic_stabilization} guarantees the error converges asymptotically to a neighborhood of the origin, and Theorem~\ref{thrm:asymptotic_tracking} guarantees that the stabilized set eventually converges to within an ultimate bound.  Theorem~\ref{thrm:asymptotic_tracking} achieves the control objective~\eqref{eq:control_objective} with \vspace{-0.17cm}
\begin{equation}
    \Delta = \gamma(\bar\g r). \vspace{-0.17cm}
\end{equation}
The size of $\Delta$ can be tuned to ensure practical utility by selecting an appropriate value of $r$, since $\gamma\in\calK_{\infty}$.   

Theorem~\ref{thrm:no_ECT} establishes the nonexistence of ECT controllers for systems of the form~\eqref{eq:olsys}.  Theorem~\ref{thrm:asymptotic_tracking} provides an affirmative answer to the question of existence of PCT controllers for the same class, formalized in the following corollary.  

\begin{corollary}
    Given Assumption~\ref{assmp:ISS_control_law}, there exists a PCT control law for all systems of the form~\eqref{eq:olsys} for any $\Delta>0$. 
\end{corollary}

\begin{proof}
    Fix $\Delta>0$ and select $r = \bar\g^{-1}\gamma^{-1}(\Delta)$.  This radius $r$ is compatible with the control law~\eqref{eq:u1_prescription},~\eqref{eq:u2_prescription} since it is positive.  Theorem~\ref{thrm:asymptotic_tracking} guarantees that applying~\eqref{eq:u1_prescription},~\eqref{eq:u2_prescription} with the selected $r$ ensures~\eqref{eq:control_objective}, the definition of PCT.  
\end{proof}

\section{SIMULATION EXAMPLE} \label{sec:Simulation}

The theoretical results are illustrated with a numerical simulation of a planar multirotor.  Physical units are omitted in the following discussion.  Let $\x:\RR_{\geq0}\to\RR^2$ denote position and $\v:\RR_{\geq0}\to\RR^2$ denote velocity.  The outer-loop state is $x(t) = [\,\x(t)\tt \,\,\, \v(t)\tt\,]\tt\in\RR^4$, and the outer loop evolves according to \vspace{-0.17cm}
\begin{equation}
    \bbm{\dot\x(t) \\ \dot\v(t)} = \underbrace{\bbm{0 & I \\ 0 & 0}\bbm{\x(t) \\ \v(t)} + \bbm{0\\ -g\e_2}}_{f(t,x(t))} + \underbrace{\bbm{0\\ I}}_{\mathclap{\g(t,x(t))}} \lambda(t) u_1(t) 
    \label{eq:example_systems} \vspace{-0.17cm}
\end{equation}
where $\lambda:\RR_{\geq0}\to\SS^1$ denotes the multirotor's attitude and $g=9.81$ represents gravity.  Let $A = \e_1\e_2\tt\otimes I$, $B = \e_2\otimes I$, $C = \e_1\tt\otimes I$, and $k = -g\, \e_2$.  The baseline controller applied to the outer loop is\vspace{-0.17cm}
\begin{equation}
    \kappa_{x\str}(t,e(t)) = (C A B)^{-1}(\ddot\x\str(t) \!-\!k\!-\! C A^2 x(t) \!-\! K e(t)) \vspace{-0.17cm}
\end{equation}
where $K=\bbm{k_\x & k_\v}\otimes I$.  With $\aa\str(t)=\kappa_{x\str}(t,e(t))$, the outer-loop evolves according to \vspace{-0.17cm}
\begin{equation}
    \dot e(t) = \bbm{0 & I \\ -k_\x I & -k_\v I}e(t) + B \Delta\a(t) \vspace{-0.17cm}
\end{equation}
which satisfies all assumptions of Section~\ref{sec:problem_formulation} for $k_\x,k_\v>0$.  

\begin{figure}[t!]
    \begin{subfigure}{0.5\textwidth}
        \centering
        \includegraphics[width=0.9\textwidth]{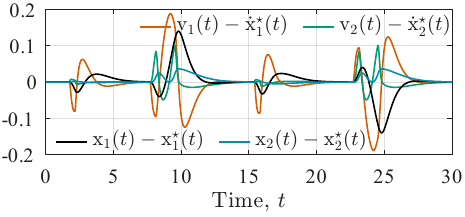}
        \vspace{-0.2cm}
        \caption{Time history of $e(t)$.}
        \label{fig:quad_error}
    \end{subfigure}\hfill 
    \begin{subfigure}{0.5\textwidth}
        \centering
        \includegraphics[width=0.9\textwidth]{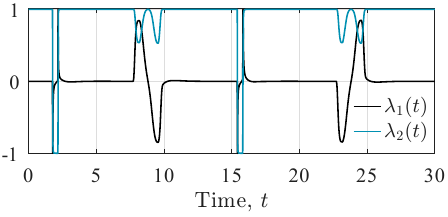}
        \vspace{-0.2cm}
        \caption{Time history of $\lambda(t)$.}
        \label{fig:quad_lambda}
    \end{subfigure}\hfill 
    \caption{Time histories from the multirotor simulation.}
    \label{fig:quad_fig}
    \vspace{-0.6cm}
\end{figure}

The simulation uses $k_\x = 3.1623$, $k_\v=4.0404$, $r=\frac{2\pi}{10}$, $\a_0=0.02$, $\a_1=0.03$, $\varrho=0.2$, $\Delta_{\a}=0.025$, $\Delta_{\dot\eta\str}=0.01$, $k_{\a}=5$, $k_{\eta}=5$, $p_{\omega}=5$, and $k_{\omega}=100$.  The target outer-loop history is $x\str(t) = [\,\x\str(t)\tt \,\,\, \dot\x\str(t)\tt\,]\tt$ wherein \vspace{-0.17cm}
\begin{equation}
    \begin{aligned}
        \x\str(t) &= 10(\zeta(t;0,3) - \zeta(t;14,17))\e_2 
        \\
        &\qquad + 10(\zeta(t;7,10) - \zeta(t;21,24))\e_1
    \end{aligned}
    \label{eq:target_motion} \vspace{-0.17cm}
\end{equation}
which traces out a square, starting from the bottom left corner and moving clockwise.  The target motion~\eqref{eq:target_motion} is intentionally aggressive to illustrate the benefits of the redesign.  

The time histories of $e(t)$ and $\lambda(t)$ for the multirotor are illustrated in Figure~\ref{fig:quad_fig}.  When the multirotor is flying along either of the vertical edges of the square, tracking the aggressive profile requires the vehicle to flip upside down to rotate its thrust vector.  In the figure, the multirotor is upside down whenever $\lambda_2(t)$ is negative.  Specifically, the aircraft flips because the thrust must compensate for the difference between $\sup_{t\in\RR_{\geq0}}\|\ddot\x\str(t)\|\approx15.75$ and $g=9.81$.  Current state-of-the-art multirotor controllers which admit the target position history fail to produce this behavior and instead encounter a division by zero on the first edge of the square due to the free-fall singularity.

\section{CONCLUSIONS} \label{sec:Conclusions}

In this paper, we presented a novel controller redesign for systems with pivoted unidirectional actuation that enables practical tracking with formal guarantees.  The results apply to a large class of systems including multirotors, rockets, hovercraft, ground vehicles, and ocean surface vehicles.  The contributions of this paper are twofold.  First, we proved the existence of output histories that cannot be tracked exactly by the considered class of systems.  Second, we proposed a controller redesign that enables all output histories to be tracked in a practical sense.  

There are many avenues for future work.  The results presented here can be extended to general gimbaled unidirectional actuation.  In practice, this extension is necessary for applications to vehicles that move in three dimensions such as multirotors.  The proposed control scheme can be modified to ensure ISpS to exogenous disturbances for the outer loop and pivot systems.  It is also worthwhile to consider the case where the unidirectional actuation is restricted to some finite interval of the positive reals.  Addressing the cases of nonzero lower-bounds and finite upper-bounds will help to close the gap between theory and practice.  Robustness guarantees in the form of strong integral ISpS guarantees can be developed for the case of finite control authority.



\begin{appendix}
    
\subsection{Smooth Step Function}
\label{appsec:zeta}

The paper uses several bespoke functions defined in this appendix.  Define $\varpi:\RR\to[0,1)$ by \vspace{-0.17cm}
\begin{equation}
    \varpi(s) = \left\{ \matl{0, & s\leq0 \\ \exp(-s^{-1}), & s>0} \right.. \vspace{-0.17cm}
\end{equation}
Note that $\varpi$ is smooth and all of its derivatives at zero are zero.  Define $\zeta_0:\RR\to[0,1]$ by \vspace{-0.17cm}
\begin{equation}
    \zeta_0(s) = \frac{\varpi(s)}{\varpi(s)+\varpi(1-s)}. \vspace{-0.17cm}
\end{equation}
This smooth step function has $\zeta_0(s)=0$ for all $s\leq0$ and $\zeta_0(s)=1$ for all $s\geq1$.  Define $\zeta:\RR\times\RR\times\RR\to[0,1]$ by \vspace{-0.17cm}
\begin{equation}
    \zeta(s;s_0,s_1) = \zeta_0\bigl(\frac{s-s_0}{s_1-s_0}\bigr). \vspace{-0.17cm}
\end{equation}
The function $\zeta(s;s_0,s_1)$ is a $C^{\infty}$ smooth step function.  Its argument is $s$ and $s_0,s_1$ are parameters. Respectively, $s_0$ and $s_1$ are the values of $s$ where the value of $\zeta$ becomes $0$ and $1$.  If $s_0=s_1$, then $\zeta$ is not defined.  If $s_0<s_1$, then $\zeta$ is zero to the left of the interval $(s_0,s_1)$ and is positive one to the right of it.  If $s_1<s_0$, then the value of $\zeta$ is positive one to the left of the interval $(s_1,s_0)$ and is zero to the right of it. It is possible to simplify equations like~\eqref{eq:defn_of_theta_star} with\vspace{-0.17cm}
\begin{equation}
    1-\zeta(s;s_0,s_1) = \zeta(s;s_1,s_0). \vspace{-0.17cm}
\end{equation}
In this paper, however, we avoid using this identity and instead always write $\zeta(s;s_0,s_1)$ with $s_0<s_1$ for clarity.

Finally, define $Z:\RR\times\RR\times\RR\to\RR_{\geq0}$ by \vspace{-0.17cm}
\begin{equation}
    Z(s;s_0,s_1) = \dintt{s_0}{s}{\zeta(\sigma;s_0,s_1)}{\sigma}. \vspace{-0.17cm}
\end{equation}
The function $Z(s;s_0,s_1)$ is a smooth ramp function based on $\zeta(s;s_0,s_1)$.  For $s_0<s_1$, $Z(s;s_0,s_1)=0$ for $s\leq s_0$ and $Z(s;s_0,s_1)=s - \tfrac{1}{2}(s_0+s_1)$ for $s\geq s_1$.

\subsection{Arctangent Identity}
\label{appsec:arctan_stuff}

\begin{lemma} \label{lemma:arctan_stuff}
    $a>b\Rightarrow a-b\geq\arctan(a)-\arctan(b)$.
\end{lemma}

\begin{proof}
    $x\mapsto x-\arctan(x)$ is nondecreasing.  Therefore\vspace{-0.17cm}
    \begin{align}
        a > b &\Rightarrow a-\arctan(a) \geq b-\arctan(b)
        \\
        &\Leftrightarrow a-b \geq \arctan(a)-\arctan(b) \\[-0.65cm] \notag
    \end{align}
    which is the claimed result.
\end{proof}

\end{appendix}


\bibliographystyle{IEEEtran}
\bibliography{root}

\end{document}